\documentclass{article}

\usepackage{microtype}
\usepackage{algorithm}
\usepackage[noend]{algpseudocode}
\usepackage{mathtools}
\usepackage{amsmath,amssymb,amsthm}
\usepackage{enumitem}

\newtheorem{theorem}{Theorem}[section]
\newtheorem{lemma}[theorem]{Lemma}
\newtheorem{definition}[theorem]{Definition}

\newcommand{\edge}[2]{#1#2}
\def\Hom{\mathsf{Hom}}
\newcommand{\BIS}[2]{\#_p BIS_{#1,#2}}
\def\Aut{\mathsf{Aut}}
\def\Fix{\mathsf{Fix}}
\let\gm=\gamma
\let\ld=\lambda
\let\om=\alpha
\let\vf=\varphi
\let\sg=\sigma
\let\dl=\delta
\let\th=\theta
\let\Dl=\Delta
\def\cG{\mathcal G}
\def\cJ{\mathcal J}
\def\cK{\mathcal K}
\let\sse=\subseteq
\def\vc#1#2{#1_1,\dots,#1_#2}
\def\ghom#1{\#\mathsf{GraphHom}(#1)}
\def\ghomk#1#2{\#_{#1}\mathsf{GraphHom}(#2)}
\def\parthom#1#2{\#_{#1}\mathsf{PartHom}(#2)}
\def\deg{\mathsf{deg}}

\newcommand{\prob}[3]{
  \vbox{
  \begin{description}
    \item[\bf Name:] #1
    \vspace{-1.75ex}
    \item[\bf Input:] #2  
    \vspace{-1.75ex}
    \item[\bf Output:] #3
  \end{description}
}}%}

\begin{document}

\title{Counting Homomorphisms Modulo a Prime Number\thanks{This work was supported by an NSERC Discovery grant}}

%% \titlerunning{Counting Homomorphisms Modulo a Prime Number}

\author{Amirhossein Kazeminia \\ Simon Fraser University \\ amirhossein.kazeminia@sfu.ca
\and
Andrei A. Bulatov \\ Simon Fraser University\\ abulatov@sfu.ca}

%% \authorrunning{A.Kazeminia and A.A.Bulatov}

%% \Copyright{Amirhossein Kazeminia and Andrei A. Bulatov}

%% \subjclass{Theory of computation $\to$  Computational complexity and cryptography $\to$  Problems, reductions and completeness }

%% \keywords{graph homomorphism, modular counting, computational hardness}

%% \acknowledgements{This work was supported by an NSERC Discovery grant}

%Editor-only macros:: begin (do not touch as author)%%%%%%%%%%%%%%%%%%%%%%%%%%%%%%%%%%
%% \EventEditors{John Q. Open and Joan R. Access}
%% \EventNoEds{2}
%% \EventLongTitle{42nd Conference on Very Important Topics (CVIT 2016)}
%% \EventShortTitle{CVIT 2016}
%% \EventAcronym{CVIT}
%% \EventYear{2016}
%% \EventDate{December 24--27, 2016}
%% \EventLocation{Little Whinging, United Kingdom}
%% \EventLogo{}
%% \SeriesVolume{42}
%% \ArticleNo{23}
%\nolinenumbers %uncomment to disable line numbering
%\hideLIPIcs  %uncomment to remove references to LIPIcs series (logo, DOI, ...), e.g. when preparing a pre-final version to be uploaded to arXiv or another public repository
%%%%%%%%%%%%%%%%%%%%%%%%%%%%%%%%%%%%%%%%%%%%%%%%%%%%%%

%% \date{January 2019}

\date{}
\maketitle
%  ABSTRACT  ======================================================================================
\begin{abstract}
Counting problems in general and counting graph homomorphisms in particular have numerous applications in combinatorics, computer science, statistical physics, and elsewhere. One of the most well studied problems in this area is $\ghom H$ --- the problem of finding the number of homomorphisms from a given graph $G$ to the graph $H$. Not only the complexity of this basic problem is known, but also of its many variants for digraphs, more general relational structures, graphs with weights, and others.
In this paper we consider a modification of $\ghom H$, the $\ghomk pH$
problem, $p$ a prime number: Given a graph $G$, find the number of homomorphisms from $G$ to $H$ modulo $p$. In a series of papers Faben and Jerrum, and G\"{o}bel et al.\ determined the complexity of $\ghomk2H$ in the case $H$ (or, in fact, a certain graph derived from $H$) is square-free, that is, does not contain a 4-cycle. Also, G\"{o}bel et al.\ found the complexity of $\ghomk pH$ for an arbitrary prime $p$ when $H$ is a tree. Here we extend the above result to show that the $\ghomk pH$ problem is $\#_p$P-hard whenever the derived graph associated with $H$ is square-free and is not a star, which completely classifies the complexity of $\ghomk pH$ for square-free graphs $H$.
\end{abstract}

%  INTRODUCTION  =================================================================================
\section{Introduction}	

% counting homs is good
A homomorphism from a graph $G$ to a graph $H$ is an edge-preserving mapping from the vertex set of $G$ to that of $H$. Graph homomorphisms provide a powerful framework to model a wide range of combinatorial problems in computer science, as well as a number of phenomena in combinatorics and graph theory, such as graph parameters \cite{Hell04:homomorphism,Lovasz12:large}. Two of the most natural problems related to graph homomorphisms is $\mathsf{GraphHom}(H)$: Given a graph $G$, decide whether there is a homomorphism from $G$ to a fixed graph $H$, and its counting version $\ghom H$ of finding the number of such homomorphisms. Special cases of these problems include the $k$-Colouring and $\#k$-Colouring problems ($H$ is a $k$-clique), Bipartiteness ($H$ is an edge), counting independent sets ($H$ is an edge with a loop at one vertex) and many others. 

% Hom(H), complexity, DG, and other results, approx
In general the $\mathsf{GraphHom}(H)$ and $\ghom H$ problems are NP-complete and \#P-complete, respectively. However, for certain graphs $H$ these problems are significantly easier. Hell and Nesetril \cite{Hell90:h-coloring} were the first to address this phenomenon in a systematic way. They proved that the $\mathsf{GraphHom}(H)$ problem is polynomial time solvable if and only if $H$ has a loop or is bipartite, and $\mathsf{GraphHom}(H)$ is NP-complete otherwise. In the counting case a similar result was obtained by Dyer and Greenhill \cite{DyerDichotomy}, in this case the $\ghom H$ problem is solvable in polynomial time if and only if $H$ a complete graph with all loops present or a complete bipartite graph, otherwise the problem is \#P-complete. This result was later generalized to computing partition functions for weighted graphs with nonnegative weights by Bulatov and Grohe \cite{Bulatov05:partition}, graphs with real weights by Goldberg et al.\ \cite{Goldberg10:mixed}, and finally for complex weights by Cai et al.\ \cite{Cai10:graph}. There have also been major attempts to find approximation algorithms for the number of homomorphisms and other related graph parameters, see, e.g., \cite{Bezakova18:inapproximability,Galanis17:complexity,Galanis16:approximately}.

% modular counting, valiant, holographic, 
The modification of the $\ghom H$ problem we consider in this paper concerns finding the number of homomorphisms modulo a natural number $k$. The corresponding problem will be denoted by $\ghomk kH$. Although modular counting has been considered by Valiant \cite{Valiant06:accidental} in the context of holographic algorithms, Faben and Jerrum \cite{Faben15:parity} where the first who systematically considered the problem $\ghomk 2H$. In particular, they posed a conjecture stating that this problem is polynomial time solvable if and only if a certain graph $H^{*2}$ derived from $H$ (to be defined later in this section) contains at most one vertex, and is complete in the class $\oplus\mathrm{P}=\#_2\mathrm{P}$ otherwise. Note that hardness results in this area usually show completeness in a complexity class $\#_k$P of counting the number of accepting paths in polynomial time nondeterministic Turing machines modulo $k$. The standard notion of reduction in this case is Turing reduction. Faben and Jerrum proved their conjecture in the case when $H$ is a tree. This result has been extended  by G\"{o}bel et al.\ first to the class of cactus graphs \cite{CountingMod2Cac} and then to square-free graphs \cite{CountingMod2SF} (a graph is a square-free if it does not contain a 4-cycle).
  
% our result
In this paper we follow the lead of G\"{o}bel et al.\ \cite{Gobel18:trees} and consider the problem $\ghomk pH$ for a prime number $p$. We only consider loopless graphs without parallel edges. There are similarities with the\!\! $\pmod2$ case. In particular, the derived graph constructed in \cite{Faben15:parity} can also be constructed following the same principles, it is denoted $H^{*p}$, and it suffices to study $\ghomk pH$ for this graph only. On the other hand, the problem is richer, as, for example, the polynomial time solvable cases include complete bipartite graphs. G\"{o}bel et al.\ \cite{Gobel18:trees} considered the case when $H$ is a tree. Recall that a \emph{star} is a complete bipartite graph of the form $K_{1,n}$. Stars are the only complete bipartite graphs that are trees. The main result of \cite{Gobel18:trees} establishes that $\ghomk pH$, $H$ is a tree, is polynomial time solvable if and only if $H^{*p}$ is a star. We generalize this result to arbitrary square-free graphs. 

\begin{theorem}\label{the:main-intro}
Let $H$ be a square-free graph and $p$ a prime number. Then the $\ghomk pH$ problem is solvable in polynomial time if and only if the graph $H^{*p}$ is a star, and is $\#_p$P-complete otherwise.
\end{theorem}

% detailed survey, methods
We now explain the main ideas behind our result, as well as, the majority of results in this area. As it was observed by Faben and Jerrum \cite{Faben15:parity}, the automorphism group $\Aut(H)$ of graph $H$ plays a very important role in solving the $\ghomk pH$ problem. Let  $\vf$ be a homomorphism from a graph $G$ to $H$. Then composing $\vf$ with an element from $\Aut(H)$ we again obtain a homomorphism from $G$ to $H$. The set of all such homomorphisms forms the orbit of $\vf$ under the action of $\Aut(H)$. If $\Aut(H)$ contains an automorphism $\pi$ of order $p$, the cardinality of the orbit of $\vf$ is divisible by $p$, unless $\pi\circ\vf=\vf$, that is, the range of $\vf$ is the set of fixed points $\Fix(\pi)$ of $\pi$ ($a\in V(H)$ is a fixed point of $\pi$ if $\pi(a)=a$). Let $H^\pi$ denote the subgraph of $H$ induced by $\Fix(\pi)$. We write $H\Rightarrow_p H'$ if there is $\pi\in\Aut(H)$ such that $H'$ is isomorphic to $H^\pi$. We also write $H\Rightarrow^*_p H'$ if there are graphs $\vc Hk$ such that $H$ is isomorphic to $H_1$, $H'$ is isomorphic to $H_k$, and $H_1\Rightarrow_p H_2\Rightarrow_p\dots\Rightarrow_p H_k$. 

\begin{lemma}[\cite{Faben15:parity}]\label{lem:aut-reduction}
Let $H$ be a graph and $p$ a prime. Up to an isomorphism there is a unique smallest (in terms of the number of vertices) graph $H^{*p}$ such that $H\Rightarrow_p^* H^{*p}$, and for any graph $G$ it holds 
\[
|\Hom(G,H)|\equiv|\Hom(G,H^{*p})|\pmod p.
\]
Moreover, $H^{*p}$ does not have automorphisms of order $p$.
\end{lemma}

The easiness part of Theorem~\ref{the:main-intro} follows from the classification of the complexity of $\ghom H$ by Dyer and Greenhill \cite{DyerDichotomy}. Since whenever $\ghom H$ is polynomial time solvable, so is $\ghomk pH$ for any $p$, Lemma~\ref{lem:aut-reduction} implies that if $H^{*p}$ is a complete graph with all loops present or a complete bipartite graph the problem $\ghomk pH$ is also solvable in polynomial time. We restrict ourselves to loopless square-free graphs, therefore, as $H^{*p}$ is isomorphic to an induced subgraph of $H$, \cite{DyerDichotomy} only guarantees polynomial time solvability when $H^{*p}$ is a star. 

Another ingredient in our result is the $\#_p$P-hard problem we reduce to $\ghomk pH$. In most of the cited works the hard problem used to prove the hardness of $\ghomk 2H$ is the problem $\#_2IS$ of finding the parity of the number of independent sets. This problem was shown to be $\#_2$P-complete by Valiant \cite{Valiant06:accidental}. We use a slightly different problem. For two positive real numbers $\ld_1,\ld_2$, let $\#_pBIS_{\ld_1,\ld_2}$ denote the following problem of counting weighted independent sets in bipartite graphs, where$\mathcal{IS}(G)$ denotes the set of all independent sets of $G$

\prob{$\BIS{\lambda_1}{\lambda_2}$}{a bipartite graph $G$}{$Z_{\lambda_1,\lambda_2}(G)=\sum_{I\in \mathcal{IS}(G)} \lambda_1 ^{|V_L \cap I|}\lambda_2 ^{{|V_R \cap I|}}\: \pmod{p}$.}

It was shown by G\"{o}bel et al.\ in \cite{Gobel18:trees} that $\#_pBIS_{\ld_1,\ld_2}$ is $\#_p$P-complete for any $\ld_1,\ld_2$, unless one of them is equal to $0\pmod p$. The main technical statement we prove here is the following

\begin{theorem}\label{the:main-technical-intro} 
Let $H$ be a square-free graph such that $H^{*p}$ is not a star. Then there are $\ld_1,\ld_2\not\equiv0\pmod p$ such that $\BIS{\ld_1}{\ld_2}$ is polynomial time reducible to the $\ghomk pH$ problem.
\end{theorem}

We note that the requirement of being square-free is present in all results on modular counting of graph homomorphism. Clearly, this is an artifact of the techniques used in all these works, and so overcoming this requirement would be a substantial achievement.

%  PRELIMINARIES  ==========================================================================================
\section{Preliminaries}

We use $[n]$ to denote the set $\{1,...,n\}$. Also, we usually abbreviate $A \setminus \{x\}$ to $A-x$. Let $k$ be a positive integer, then for a function $f$ its $k$-fold composition is denoted by $f^{(k)}=f \circ f \circ \cdots \circ f$.

\textbf{Graphs.}
In this paper, \textit{graphs} are undirected, and have no parallel edges or loops. For a graph $G$, the set of vertices of $G$ is denoted by $V(G)$, and the set of edges is denoted by $E(G)$. We use $\edge{u}{v}$ to denote an edge of $G$.  
The set of \emph{neighbours} of a vertex $v\in V(G)$ is denoted by $N_G (v)= \{ u \in V(G) : \edge{u}{v} \in E(G) \}$, and the degree of $v$ is denoted by $\deg(v)$. 
%% The length of a shortest path between vertices $u$ and $v$ is denoted by $d(u,v)$.

A set $I \subseteq V(G)$ is an \emph{independent set} of $G$ if and only if $\edge{u}{v}$ is an edge of $G$ for no $u,v \in I$. The set of all independent sets of $G$ is denoted by $\mathcal{IS}(G)$. If $G$ is a bipartite graph, the parts of a bipartition of $V(G)$ will be denoted $V_R(G)$ and $V_L(G)$ in no particular order.

\textbf{Homomorphisms.}
A \emph{homomorphism} from a graph $G$ to a graph $H$ is a mapping $\vf$ from $V(G)$ to $V(H)$ which preserves edges, i.e.\ for any $\edge uv\in E(G)$ the pair $\edge{\vf(u)}{\vf(v)}$ is an edge of $H$. The set of all homomorphisms from $G$ to $H$, is denoted by $\Hom(G,H)$. For a graph $H$ the problem of counting homomorphisms from a graph $G$ to $H$ is denoted by $\ghom H$. The problem of finding the number of homomorphisms from a given graph $G$ to $H$ modulo $k$ is denoted by $\ghomk kH$:

\prob{$\ghomk kH$}{ a graph $G$}{$|\Hom(G,H)| \pmod{k}$.}

It will be convenient to denote the vertices of the graph $H$ by lowercase Greek letters. 

A homomorphism $\vf$ from $G$ to $H$ is an \textit{isomorphism} if it is bijective and for all $u,v \in V(G)$, $\edge{u}{v} \in E(G)$ if and only if $\edge{\vf(u)}{\vf(v)}\in E(H)$.  An \textit{automorphism} of $G$ is an isomorphism from graph $G$ to itself. The automorphism group of $G$ is denoted by $\Aut(G)$. An automorphism $\pi$ is an \textit{automorphism of order $k$} if $k$ is the smallest positive integer such that $\pi^{(k)}$ is the identity transformation. A \textit{fixed point} of an automorphism $\pi$ of $G$ is a vertex $v\in V(G)$ such that $v=\pi(v)$.  

\textbf{Partially labelled graphs.}
A \emph{partial function} from $X$ to $Y$ is a function $f:X' \rightarrow Y$ for a subset $X' \subseteq X$. For a graph $H$, a \emph{partial $H$-labelled graph} $\cG$ is a graph $G$ (called the \emph{underlying graph} of $\cG$) equipped with a \emph{pinning function} $\tau$, which is a partial function from $V(G)$ to $V(H)$. %% We sometimes denote the underlying graph of $\cG$ by $G(\cG)$ and the pinning function by $\tau(\cG)$.
A homomorphism from a partial $H$-labelled graph $\cG=(G,\tau)$ to a graph $H$ is a homomorphism $\sigma : G \rightarrow H$ that extends the pinning function $\tau$, that is, for all $v\in \text{dom}(\tau)$, $\sigma (v)= \tau (v)$. The set of all such homomorphisms is denoted by $\Hom(\cG,H)$.

In certain situations it will be convenient to use a slightly different view on collections of homomorphisms of $H$-labelled graphs. A set of homomorphisms $\vf$ from a graph $G$ to $H$ that map vertices $x_1,x_2,...,x_r \in V(G)$ to vertices $y_1,y_2,...,y_r \in V(H)$ such that $\vf(x_i)=y_i$ for $i\in [r]$ is denoted by $\Hom((G,x_1,x_2,...,x_r),(H,y_1,y_2,...,y_r)).$

\textbf{Counting complexity classes.}
The class $\#$P is defined to be the class of problems of counting the accepting paths of a polynomial time nondeterministic Turing machine. This means every problem in NP has an associated counting problem in \#P, so for $A\in\mathrm{NP}$, an associated counting problem will be denoted by $\#A$. (Strictly speaking for every such problem the corresponding counting one is not uniquely defined, but in our case there will always be the `natural' one.) Classes $\#_k$P, where $k$ is a natural number are defined in a similar way, as counting the accepting paths in a polynomial time nondeterministic Turing machine modulo $k$. For $A\in\mathrm{NP}$ the corresponding problem in $\#_k\mathrm{P}$ is denoted by $\#_kA$.

Several kinds of reductions between counting problems have appeared in the literature. The first one, parsimonious, was introduced in the foundational papers \cite{Valiant79:permanent,Valiant79:complexity} by Valiant. A counting problem $A$ is \emph{parsimoniously reducible} to a counting problem $B$, denoted $A\le B$, if there is a polynomial time algorithm that, given an instance $I$ of $A$, produces an instance $J$ of $B$ such that the answers to $I$ and $J$ are the same. The other type of reduction frequently used for counting problems is Turing reduction. Counting problem $A$ is \emph{Turing reducible} to problem $B$, denoted $A\le_T B$, if there exists a polynomial time algorithm solving $A$ and  using $B$ as an oracle. 

These two types of reductions can be applied to modular counting as well. Turing reduction does not require any modifications. For parsimonious reduction we say that a problem $A$ from $\#_k\mathrm{P}$ is \emph{parsimoniously reducible} to a problem $B$ from $\#_k\mathrm{P}$ if there is a polynomial time algorithm that, given an instance $I$ of $A$, produces an instance $J$ of $B$ such that the answers to $I$ and $J$ are congruent modulo $k$. In this paper we mostly claim Turing reducibility, although our main technical result constructs a parsimonious reduction. However, the proof of Theorem~\ref{the:main-intro} involves other reductions that are not always parsimonious. Problem $\#_kA$ is said to be \emph{$\#_k\mathrm{P}$-complete} if it belongs to $\#_k\mathrm{P}$ and every problem from $\#_k\mathrm{P}$ is Turing reducible to $\#_kA$.

% MODULAR COUNTING COMPLEXITY   ===========================================================================

\section{Outline of the proof}

In this section we outline our proof strategy and formally introduce all the necessary intermediate problems and existing results.  Fix a prime number $p$.

As it was observed in the introduction, Lemma~\ref{lem:aut-reduction} proved by  Faben and Jerrum \cite{Faben15:parity} combined with the classification by Dyer and Greenhill \cite{DyerDichotomy} proves the easiness part of Theorem~\ref{the:main-intro}. We therefore focus on proving the hardness part. Again, by Lemma~\ref{lem:aut-reduction} we may assume that $H$ does not have automorphisms of order $p$.

For the hardness part, we use two auxiliary problems. The first one is the problem $\BIS{\lambda_1}{\lambda_2}$ mentioned in the introduction. 
Let $\lambda_1 , \lambda_2 \in\{0,\dots p-1\}$, and let $G=(V_L \cup V_R, E)$ be a bipartite graph. Define the following weighted sum over independent sets of $G$: 
\begin{equation*}
    Z_{\lambda_1,\lambda_2}(G)=\sum_{I\in \mathcal{IS}(G)} \lambda_1 ^{|V_L \cap I|}\lambda_2 ^{{|V_R \cap I|}}.
\end{equation*}
The problem of computing function $Z_{\lambda_1,\lambda_2}(G)$ for a given bipartite graph $G$, prime number $p$ and $\lambda_1, \lambda_2 \in\{0,\dots p-1\}$,  is defined as follows:

\prob{$\BIS{\lambda_1}{\lambda_2}$}{a bipartite graph $G$}{$Z_{\lambda_1,\lambda_2}(G)\: \pmod{p}$.}

The complexity of $\BIS{\lambda_1}{\lambda_2}$ was determined by G\"{o}bel, Lagodzinski and Seidel \cite{Gobel18:trees}. 

\begin{theorem}\label{TH:weightedCompleteness}\cite{Gobel18:trees}
If $\lambda_1 \equiv 0 \pmod{p}$ or $\lambda_2 \equiv 0 \pmod{p}$ then the problem $\BIS{\lambda_1}{\lambda_2}$ is solvable in polynomial time, otherwise it is $\#_p\mathrm{P}$-complete.
\end{theorem}

The second auxiliary problem has been used in all works on $\ghomk pH$ starting from the initial paper by Faben and Jerrum \cite{Faben15:parity}. It is the problem of counting homomorphisms from a given partially $H$-labelled graph $\cG$ to a fixed graph $H$ modulo prime $p$.

\prob{$\parthom pH$}{a partial $H$-labelled graph $\cG=(G,\tau)$}{$|\Hom(\cG,H)| \pmod{p}$.}

The chain of reductions we use to prove the hardness part of Theorem~\ref{the:main-intro} is the following:
\begin{equation}\label{RED:chain.}
    \BIS{\lambda_1}{\lambda_2} \leq_T \parthom p{H^{*p}} \leq_T \ghomk p{H^{*p}} \le_T \ghomk pH.
\end{equation}

The last reduction is by Lemma~\ref{lem:aut-reduction}. Acually, the two last problems in the chain are polynomial time interreducible through (modular) parsimonious reduction. The second step, the reduction from $\parthom pH$ to $\ghomk pH$ was proved by G\"{o}bel, Lagodzinski and Seidel \cite{Gobel18:trees}.

\begin{theorem}\label{TH:SECRED}\cite{Gobel18:trees}
Let $p$ be a prime number and let $H$ be a graph that does not have any automorphism of order $p$. Then $\parthom pH$ can be reduced to $\ghomk pH$ through a polynomial time Turing reduction.
%% \begin{equation*}
%% \parthom pH \leq_T \ghomk pH.
%% \end{equation*}
\end{theorem}

Finally, the first reduction in the chain is our main contribution. We show it in three steps. Recall that we are reducing the problem of finding the number of (weighted) independent sets in a bipartite graph to the problem of finding the number of extensions of a partial homomorphism from a given graph to $H$. First, in Section~\ref{sec:properties} starting from a bipartite graph $G$ we replace its vertices and edges with gadgets, whose exact structure we do not specify at that point. We call those gadget the \emph{vertex} and \emph{edge} gadgets. Then we show that if the vertex and edge gadgets satisfy certain conditions, in terms of the number of homomorphisms of a certain kind from the gadgets to $H$ (Theorem~\ref{TH:hardness}), then $Z_{\lambda_1,\lambda_2}(G)\: \pmod{p}$ can be found in polynomial time from $|\Hom(\cG,H)| \pmod{p}$, where $\cG$ is the partially $H$-labelled graph constructed in the reduction. In the second step, Section~\ref{sec:gadgets}, we introduce several variants of vertex and edge gadgets and show some of their properties. Finally, in Section~\ref{sec:cases} we consider several cases depending on the degree sequence of the graph $H$. In every case we construct vertex and edge gadgets that satisfy the conditions of Theorem~\ref{TH:hardness}, thus completing the reduction.

% HARDNESS FOR PARTIAL H-LABELLED GRAPHS  =================================================================
\section{Hardness gadgets}\label{sec:properties}

Our goal in this section is to describe a general scheme of a reduction from $\BIS{\lambda_1}{\lambda_2}$ to $\parthom pH$, where $p$ is prime and $H$ is a square-free graph. 

The general idea is, given a bipartite graph $G=(V_L\cup V_R,E)$, where $V_L,V_R$ is the bipartition of $G$, to construct a new partially $H$-labelled graph $\cG'$, which is obtained from $G$ by adding a copy of a \emph{vertex} gadget $\cJ$ to every vertex of $G$, and replacing every edge from $E$ with a copy of an \emph{edge} gadget $\cK$. The gadgets are partially $H$-labelled graphs and their pinning functions will define the pinning function of $\cG'$. Since $G$ is a bipartite graph, the vertex gadget comes in two versions, left, $\cJ_L$, and right, $\cJ_R$. Also, both vertex gadgets have a distinguished vertex, $s$ for $\cJ_L$ and $t$ for $\cJ_R$. The edge gadget $\cK$ has two distinguished vertices, $s$ and $t$. These distinguished vertices will be identified with the vertices of the original graph $G$, as shown in Fig.~\ref{fig:reduction-scheme}.

\begin{figure}[ht]
\centering
\includegraphics[totalheight=4cm,keepaspectratio]{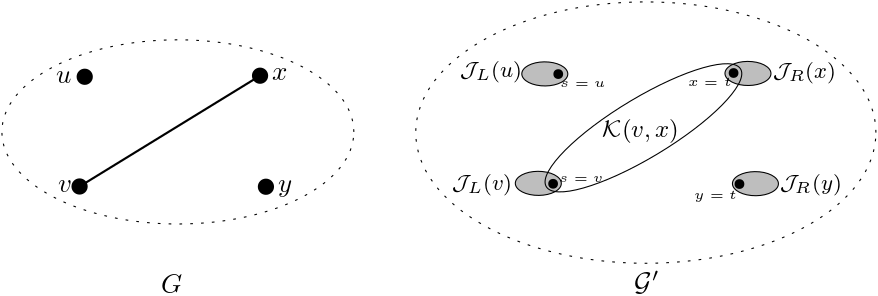}
\caption{The structure of graph $\cG'$. The original graph $G$ is on the left. The resulting graph $\cG'$ is on the right: vertex gadgets $\cJ_L,\cJ_R$ are added to every vertex, and the only edge $vx$ of $G$ is replaced with a copy of gadget $\cK$.}
\label{fig:reduction-scheme}
\end{figure}

The gadgets $\cJ_L,\cJ_R$ are associated with sets $\Dl_1,\Dl_2\sse V(H)$ and vertices $\dl_1\in\Dl_1,\dl_2\in\Dl_2$, respectively. The pinning functions of $\cJ_L,\cJ_R$ will be defined in such a way that for any homomorphism $\vf$ of $\cJ_L$ ($\cJ_R$) to $H$, vertex $s$ (respectively, $t$) is forced to be mapped to $\Dl_1$ (respectively, $\Dl_2$). For $x\in V_L$ let $\cJ_L(x)$ denote the copy of $\cJ_L$ connected to $x$, that is, $s$ in $\cJ_L(x)$ is identified with $x$. For $y\in V_R$ the copy $\cJ_R(y)$ is defined in the same way. The vertices $\dl_1,\dl_2$ will help to encode independent sets of $G$. Specifically, with every independent set $I$ of $G$ we will associate a set of homomorphisms $\vf:\cG'\to H$ such that for every vertex $x\in V_L$, $x\in I$ if and only if $\vf(x)\ne\dl_1$ (recall that $x$ is also a vertex of $\cG'$ identified with $s$ in $\cJ_L(x)$); and similarly, for every $y\in V_R$, $y\in I$ if and only if $\vf(y)\ne\dl_2$. Finally, the edge gadgets $\cK(x,y)$ replacing every edge $xy\in E$ make sure that every homomorphism from $\cG'$ to $H$ is associated with an independent set.

Note that just an association of independent sets with collections of homomorphisms is not enough, the number of homomorphisms in those collections have to allow one to compute the function $Z_{\lambda_1, \lambda_2}(G)$.

Next we introduce conditions such that  if for the graph $H$ there are vertex and edge gadgets satisfying these conditions, then $\BIS{\lambda_1}{\lambda_2}$ for some nonzero (modulo $p$) $\ld_1,\ld_2$ is reducible to $\parthom pH$.

\begin{definition}[Hardness gadget]\label{def::hardness}
A graph $H$ has \emph{hardness gadgets} if there are $\Delta_1,\Delta_2 \sse V(H)$, vertices $\delta_1 \in \Delta_1$ and $\delta_2 \in \Delta_2$, and three partially $H$-labelled graphs $\cJ_L$, $\cJ_R$, and $\cK$ that satisfy the following properties:
\begin{enumerate}[label=(\roman*)]
\item \label{p:size}
$|\Delta_1|-1 \not \equiv 0 \pmod{p}$ , $|\Delta_2|-1 \not \equiv 0 \pmod{p}$;
\item \label{p:expansion}
for any homomorphism $\sigma:\cJ_L\to H$ ($\sg:\cJ_R\to H$) it holds that $\sg(s)\in\Dl_1$ (respectively, $\sg(t)\in\Dl_2$); for any homomorphism $\sigma:\cK\to H$ it holds that $\sg(s)\in\Dl_1, \sg(t)\in\Dl_2$;
\item \label{p:node_map}
for any $\gamma_1 \in \Delta_1, \gm_2\in\Dl_2$, it holds 
\[
|\Hom((\cJ_L,s),(H,\gamma_1))| \equiv |\Hom((\cJ_R,t),(H,\gamma_2))|\equiv1 \pmod{p}, 
\]
and for any $\gamma_1\not \in \Delta_1,\gm_2\not\in\Dl_2$, it holds 
\[
\Hom((\cJ_L,s),(H,\gamma_1)) = \Hom((\cJ_L,s),(H,\gamma_1))=\emptyset;
\]
\item \label{p:edge_map1}
for any  $\om_1 \in \Delta_1-\delta_1, \om_2\in \Delta_2-\delta_2$, it holds $\Hom((\cK,s,t),(H,\om_1,\om_2)) = \emptyset$;
\item \label{p:edge_map2}
for any  $\om_1 \in \Delta_1-\delta_1$, it holds $|\Hom((\cK,s,t),(H,\om_1,\delta_2))|\equiv 1 \pmod{p}$;
\item \label{p:edge_map3}
for any  $\om_2 \in \Delta_2-\delta_2$, it holds $|\Hom((\cK,s,t),(H,\delta_1,\om_2))|\equiv 1 \pmod{p}$;
\item \label{p:edge_map4}
$|\Hom((\cK,s,t),(H,\delta_1,\delta_2))|\equiv 1 \pmod{p}$.
\end{enumerate}
\end{definition}

Now we are ready to state the main result of this section.

\begin{theorem}\label{TH:hardness}
If $H$ has hardness gadgets, then for some $\ld_1,\ld_2\not\equiv0\pmod p$ the problem $\BIS{\lambda_1}{\lambda_2}$ is polynomial time reducible to $\parthom pH$. In particular, $\parthom pH$ is $\#_p$P-complete.
\end{theorem}

\begin{proof}
Let $\cJ_L,\cJ_R,\cK$ be the collection of gadgets whose existence is the  assumption of the theorem. Let also $\Dl_1,\Dl_2\sse V(H)$ and $\dl_1\in\Dl_1,\dl_2\in\Dl_2$ be sets and elements associated with the gadgets. Recall that we assume the existence of distinguished elements $s,t$ in the gadgets. Let $G=(V_L \cup V_R, E)$ be a bipartite graph. We give a detailed construction of a partially $H$-labelled graph $\cG'$, see also Fig.~\ref{fig:reduction-scheme}. 

\begin{itemize}
\item
The vertex set of $\cG'$ consists of a disjoint copy of $\cJ_L(x)$ for each $x\in V_L$, a disjoint copy of $\cJ_R(y)$ for each $y\in V_R$; distinguished vertices $s,t$ of the gadgets are identified with $x$ and $y$, respectively. Also, the vertex set includes a disjoint copy of $\cK(x,y)$ for each edge $xy\in E$. Again the distinguished vertices $s,t$ of $\cK(x,y)$ are identified with $x$ and $y$, respectively. More formally,
\[
V(\cG')= \Big( \bigcup_{x \in V_L} V(\cJ_L(x)) \Big) \cup \Big( \bigcup_{y \in V_R} V(\cJ_R(y)) \Big) \cup
    \Big( \bigcup_{\edge{x}{y} \in E} V(\cK(x,y)) \Big).
\]
\item
The edge set of $\cG'$ consists of a disjoint copy of the edge set of $\cJ_L(x)$ for each $x\in V_L$, a disjoint copy of the edge set of $\cJ_R(y)$ for each $y\in V_R$, and a disjoint copy of $\cK(x,y)$ for each edge $xy\in E$. More formally,
\[
E(\cG')= \Big( \bigcup_{x \in V_L} E(\cJ_L(x)) \Big) \cup \Big( \bigcup_{y \in V_R} E(\cJ_R(y)) \Big) \cup \Big( \bigcup_{\edge{x}{y} \in E} E(\cK(x,y)) \Big).
\]
\item
The pinning function $\tau$ of $\cG'$ is defined to be the union of the pinning functions of all the gadgets involved: function $\tau_x$, $x\in V_L$, for $\cJ_L(x)$, function $\tau_y$, $y\in V_R$, for $\cJ_R(y)$, and function $\tau_{xy}$, $\edge xy\in E$, for $\cK(x,y)$. More formally, 
\[
\tau= \Big( \bigcup_{x \in V_L} \tau_x \Big)\cup \Big( \bigcup_{y \in V_R} \tau_y \Big) \cup \Big( \bigcup_{\edge{x}{y} \in E} \tau_{xy}\Big).
\]
\end{itemize}

Let us set $\ld_1=|\Dl_1|-1, \ld_2=|\Dl_2|-1$. We now proceed to showing that $Z_{\lambda_1,\lambda_2}(G) \equiv |\Hom(\cG',H)| \pmod{p}$.  First, we show that every homomorphism corresponds to an independent set.

For each $\sigma \in \Hom(\cG',H)$, define  
\begin{equation*}
    \chi_{\sigma}=\{ x \in V_L : \sigma(x)\ne\delta_1 \} \cup\{ y \in V_R : \sigma(y)\ne \delta_2 \}.
\end{equation*}
We claim that $\chi_{\sigma}$ is an independent set. Indeed, assume that for some $\sigma \in\Hom(\cG',H)$ the set $\chi_{\sigma}$ is not an independent set in $G$, i.e.\  there are two vertices $a,b\in \chi_{\sigma}$ such that $\edge{a}{b}\in E(G)$. Without loss of generality, let $a\in V_L$ and $b \in V_R$. 
By the construction of $\chi_{\sigma}$, $\sigma(a)\neq \delta_1$ and $\sigma(b)\neq \delta_2$, by Definition \ref{def::hardness}\ref{p:expansion} we have $\sigma(a)\in \Delta_1 -\delta_1$ and $\sigma(b)\in \Delta_2- \delta_2$. Then by Definition \ref{def::hardness}\ref{p:edge_map1} the set $\Hom(\cK(a,b),a,b),(H,\om_1,\om_2))$ is empty, that is $\sg$ is not a homomorphism. A contradiction.

Let $\sim_\chi$ be a relation on $\Hom(\cG',H)$ given by 
$\sigma \sim_{\chi} \sigma'$ if and only if $\chi_{\sigma} = \chi_{\sigma '}$.
Obviously $\sim_\chi$ is an equivalence relation on $\Hom(\cG',H)$. We denote the class of\linebreak $\Hom(\cG',H) \mathbin{/} \sim_{\chi}$ containing $\sigma$ by $[\sigma]$. Clearly, the $\sim_\chi$-classes correspond to independent sets of $G$. We will need the corresponding mapping
\[
\mathfrak{F}:\Hom(\cG',H) \mathbin{/} \sim_{\chi} \longrightarrow \mathcal{IS}(G),
\quad\text{where}\quad \mathfrak{F}([\sigma])=\chi_{\sigma}
\]

First, we will prove that $\mathfrak{F}$ is bijective, then compute the cardinalities of classes $[\sigma]$.

\smallskip

\noindent
{\sc Claim 1.}
The function $\mathfrak{F}$ is bijective.

\begin{proof}[Proof of Claim 1]
By the definition of $\mathfrak{F}$, it is injective. To show surjectivity 
let $I\in \mathcal{IS}(G)$ be an independent set. We construct a homomorphism $\sigma\in\Hom(\cG',H)$ such that $\chi_{\sigma}=I$:
 
For every vertex $x \in I \cap V_L$, pick a vertex $\gamma^x_I \in \Delta_1-\delta_1$ and set $\sigma(x)=\gamma_I^x$.  For every vertex $y\in N_G(x)$, set $\sigma(y)=\delta_2$. For every vertex $x' \in V_L \setminus I$ set $\sigma(x')=\delta_1$. For every vertex $y \in I \cap V_R$, pick a vertex $\omega^y_I \in \Delta_2-\delta_2$ and set $\sigma(y)=\omega_I^y$. Note that in this case the value of $\sg(y)$ is not yet set, because $y\in N_G(x)$ for no $x\in I$. Finally, for every vertex $y' \in V_R \setminus I$ set $\sigma(y')=\delta_2$.
 
As $I$ is an independent set, for any $v\in I$ and $u\in N_G(v)$ we have $u \not \in I$. By construction of $\sigma$, if $\edge{x}{y}\in E(G)$ and $\sigma(x)=\gamma^x_I$ then $\sigma(y)=\delta_2$. Similarly, if $\sg(y)=\omega^y_I$ then $x\in N_G(y)$, and so $\sg(x)=\dl_1$. If none of the endpoints of an edge $\edge{x}{y}$ belongs to $I$ then $\sigma(x)=\delta_1$ and $\sigma(y)=\delta_2$. By Definition \ref{def::hardness}\ref{p:edge_map1},\ref{p:edge_map2} and \ref{p:edge_map4} $\sigma$ can be extended to a homomorphism from $\cG'$ to $H$. Hence $\mathfrak{F}$ is surjective.
\end{proof}

\smallskip
\noindent
{\sc Claim 2.} 
$|[\sigma]|\equiv(|\Delta_1|-1)^{|V_L \cap \chi_{\sigma}|} (|\Delta_2|-1)^{|V_R \cap \chi_{\sigma}|} \pmod{p}$.

\begin{proof}[Proof of Claim 2]
Is suffices to count the number of homomorphisms $\sigma' \in [\sigma]$. Since $\chi_{\sg'}=\chi_\sg$, for every $x\not\in I$ the value $\sg'(x)$ equals $\dl_1$ or $\dl_2$ depending on whether $x\in V_L$ or $x\in V_R$. As we have shown in Claim~1, for every vertex $x \in I \cap V_L$, $\sg'(x)\in \Delta_1-\delta_1$, so there are $|\Delta_1|-1$ options for $\sg'(x)$. Similarly, there are $|\Dl_2|-1$ options for $\sg'(y)$ for every $y\in I\cap V_R$. Therefore
 \begin{align*}\label{eq::classCard}
     |[\sigma]|&=(|\Delta_1|-1)^{|V_L \cap \chi_{\sigma}|} (|\Delta_2|-1)^{|V_R \cap\chi_{\sigma}|} \Big( \prod_{x\in V_L} |\Hom(\cJ_L(x),x),(H,\sg'(x))| \Big)\\
     &\qquad\times \Big( \prod_{y\in V_R} |\Hom(\cJ_R(y),y),(H,\sg'(y))| \Big) \\
     &\qquad\times\Big( \prod_{xy\in E} |\Hom(\cK(x,y),x,y),(H,\sg'(x),\sg'(y))| \Big) \\
     &\equiv (|\Delta_1|-1)^{|V_L \cap \chi_{\sigma}|} (|\Delta_2|-1)^{|V_R \cap \chi_{\sigma}|} \pmod{p}.
 \end{align*}
 
Where the second equality is by Definition \ref{def::hardness}\ref{p:node_map}, \ref{p:edge_map2}, \ref{p:edge_map3}, \ref{p:edge_map4}.
\end{proof}

Assume that $\sim_{\chi}$ has $M$ classes and $\sg_i$ is a representative of the $i$-th class. Then
\begin{align*}
    |\Hom(\cG',H)|&=\sum_{i=1}^M |[\sigma_i]|\\
                &\equiv \sum_{i=1}^M (|\Delta_1|-1)^{|V_L \cap \chi_{\sigma_i}|} (|\Delta_2|-1)^{{|V_R \cap \chi_{\sigma_i}|}}\\
                &\equiv \sum_{I \in \mathcal{IS}(G)} (|\Delta_1|-1)^{|V_L \cap I|} (|\Delta_2|-1)^{{|V_R \cap I|}}\\
                &\equiv Z_{|\Delta_1|-1,|\Delta_2|-1}(G)\pmod p.
\end{align*}

Therefore $\BIS{\lambda_1}{\lambda_2} \leq_T \parthom pH$. By Definition \ref{def::hardness}\ref{p:size} $\ld_1=|\Delta_1|-1, \ld_2=|\Delta_2|-1 \not \equiv 0 \pmod{p}$. As $\BIS{\ld_1}{\ld_2}$ is $\#_p\mathrm{P}$-complete by Theorem \ref{TH:weightedCompleteness}, so is $\parthom pH$.
\end{proof}

% HARDNESS GADGETS  ==============================================================================
\section{Hardness gadgets and nc-walks}\label{sec:gadgets}

In this section we make the next iteration in constructing hardness gadgets and give a generic structure of such gadgets that will later be adapted to specific types of the graph $H$. 

These gadgets make use of the square-freeness of graph $H$ that we will apply in the following form.

\begin{lemma}\label{lem:square}
Let $H$ be a square-free graph. Then for any $\alpha,\beta \in H$,  $|N_H(\alpha) \cap N_H(\beta)| \leq 1$. 
\end{lemma}

\begin{proof}
If there are two different elements $\gm,\dl$ in $N_H(\alpha) \cap N_H(\beta)$, then $\alpha,\gm,\beta,\dl$ form a 4-cycle.
\end{proof}

We call a walk in $H$ a \emph{non-consecutive-walk} or \emph{nc-walk}, if it does not traverse an edge forth and them immediately back. More formally, an nc-walk is a walk $v_0,v_1,\dots,v_k$ such that for no $i\in[k-1]$ we have $v_{i-1}=v_{i+1}$.

%%%%%%%%%%%%%%%%%%%%%%%%%%%%%%%%%%
\subsection{Edge gadget}\label{sec:edge-gadget}

Let $W=\gamma_0 \gamma_1 \cdots \gamma_k$ be an nc-walk in $H$ of length at least one. Then the edge gadget $\cK$ is a path $sv_1 v_2  \cdots  v_{k-1}t$, where each $v_i$ is connected to another vertex $u_i$ which is pinned to $\gamma_i$. More formally,
%% 
%% \begin{definition}
%% For a nc-walk $W=\gamma_0 \gamma_1 \cdots \gamma_k$ in $H$ of length at least one, a 
the gadget $\cK=(K,\tau)$ is defined as follows
\begin{gather*}
    V(K)=\{ s,t \} \cup \{ v_i, u_i : i \in [k-1] \}, \\
    E(K)= \{ v_iv_{i+1}: i \in [k-2] \} \cup \{v_iu_i : i \in [k-1] \} \cup \{sv_1,v_{k-1}t\}.
\end{gather*}
The pinning function is $\tau (u_i)=\gamma_i$ for all $i \in [k-1]$.
%% \end{definition}

The next two lemmas give some of the properties listed in Definition~\ref{def::hardness}.

\begin{lemma}[\textbf{Shifting}]\label{lem:shiftingWalk}
Let $H$, $W=\gamma_0 \gamma_1 \cdots \gamma_k$, and $\cK$ be as above. Then 
\begin{enumerate}
\item[(1)] 
For every $\theta \in N_H(\gamma_0)-\gamma_1$ and $\sigma \in\Hom((\cK,s),(H,\theta))$, we have $\sigma(v_{i})=\gamma_{i-1}$ for all $i \in [k-1]$.
\item[(2)]
For every $\theta \in N_H(\gamma_k)-\gm_{k-1}$ and $\sigma \in\Hom((\cK,t),(H,\theta))$, we have $\sigma(v_{i})=\gamma_{i+1}$ for all $i \in [k-1]$.
\end{enumerate}
\end{lemma}

\begin{proof}
If $k=1$, then both cases are trivial. We prove item $(1)$ by induction on $j\in [k-1]$, item (2) can be proved using $\gamma_k \gamma_{k-1} \cdots \gamma_0$ instead of $\gamma_0 \gamma_1 \cdots \gamma_k$.

For $j=1$, the vertex $v_1$ must be mapped to a common neighbour of $\theta$ and $\gamma_1$ because $\tau(u_1)=\gamma_1$. It means $\sigma(v_1) \in N_H(\theta) \cap N_H(\gamma_1)=\{\gm_0\}$, because $\gamma_0 \in N_H(\theta) \cap N_H(\gamma_1)$ and $H$ is a square-free graph.

Now assume that $\sigma(v_{j-1})=\gamma_{j-2}$. Similar to the base case, $\sigma(v_j) \in N_H(\gamma_{j-2}) \cap N_H(\gamma_j)$. By the same argument, the only member of this intersection is $\gamma_{j-1}$. Thus, $\sigma(v_j)=\gamma_{j-1}$.
\end{proof}

\begin{lemma}[\textbf{Counting}]\label{lem:coungingWalk}
Let $H$ be a square-free graph and let $W=\gamma_0 \gamma_1 \cdots \gamma_k$, $k\ge1$ be an nc-walk in $H$. For any $\om_s \in N_H (\gamma_0)-\gamma_1$ and $\om_t \in N_H (\gamma_k)-\gamma_{k-1}$ the following equalities hold
\begin{enumerate}
\item [(1)]
    $|\Hom((\cK,s,t),(H,\om_s,\om_t))|=0$,
\item [(2)]
    $|\Hom((\cK,s,t),(H,\gamma_1,\om_t))|=1$,
\item [(3)]
    $|\Hom((\cK,s,t),(H,\om_s,\gamma_{k-1}))|=1$,
\item [(4)]
    $ \displaystyle |\Hom((\cK,s,t),(H,\gamma_1,\gamma_{k-1}))| = 1 + \sum\limits_{i=1}^{k-1} (\deg(\gamma_i)-1)$.
\end{enumerate}
\end{lemma}

\begin{proof}
For item (1), suppose towards contradiction that there is\linebreak $\sigma \in\Hom((\cK,s,t),(H,\om_s,\om_t))$. Then it implies $\sigma \in\Hom((\cK,s),(H,\om_s))$. Therefore by Lemma \ref{lem:shiftingWalk} $\sigma(v_i)=\gamma_{i-1}$ for all $i \in [k-1]$. We also have $\sigma \in\Hom((\cK,t),(H,\om_t))$. Hence by Lemma~\ref{lem:shiftingWalk} $\sigma(v_{i})=\gamma_{i+1}$ for all $i \in [k-1]$, a contradiction.

For item (2), let $\sigma \in\Hom((\cK,s,t),(H,\gamma_1,\om_t))$. Then $\sigma \in\Hom((\cK,t),(H,\om_t))$, and by Lemma \ref{lem:shiftingWalk} all the values of $\sg$ are uniquely determined, that is, there is only one such $\sigma$. So $|\Hom((\cK,s,t),(H,\gamma_1,\om_t))|=1$. The symmetric argument works for item (3).

For item (4), note that there is a homomorphism $\sigma_0 \in\Hom((\cK,s,t),(H,\gamma_1,\gamma_{k-1}))$ such that $\sigma_0(v_j)=\gm_{j+1}$ for all $j \in [k-1]$. Suppose a homomorphism $\sigma \in\Hom((\cK,s,t),(H,\gamma_1,\gamma_{k-1}))$ is such that $\sigma_0 \neq \sigma$, i.e.\ for some $j\in [k-1]$ it holds that $\sigma(v_j)\neq \gamma_{j+1}$. Let $j$ be the smallest such $j$. Then for every $i>j$, $\sigma(v_i)=\gm_{i-1}$ is determined uniquely by Lemma \ref{lem:shiftingWalk} applied to the walk $\gm_j,\dots\gm_{k-1}t$ and path $v_j,\dots,v_k$. Also we assumed, for all $i<j$, that $\sigma(v_i)=\gm_{i+1}$. The remaining case is $i=j$, in which the image of $\sigma(v_j)$ can be chosen from $N_H(\gamma_j)-\gamma_{j+1}$. Thus, there are $\deg(\gamma_j)-1$ options. Hence, the number of homomorphisms $\sg$ such that $\sigma(v_j)\neq \gamma_{j+1}$ is $\deg(\gamma_j)-1$. Finally, 
\begin{equation*}
|\Hom((\cK,s,t),(H,\gamma_1,\gamma_{k-1}))|= 1 + \sum_{i=1}^{k-1} (\deg(\gamma_i)-1).
\end{equation*}
\end{proof}

%%%%%%%%%%%%%%%%%%%%%%%%%%%%%%%%%%%
\subsection{Vertex gadgets}\label{sec:node-gadget}
In this section we construct a vertex gadget. The main role of these gadgets is to restrict the possible images of the designated vertices $s$ and $t$ as required in Definition~\ref{def::hardness}\ref{p:expansion}, and then do it in such a way that property \ref{p:node_map} in Definition~\ref{def::hardness} is also satisfied. We present vertex gadgets of two types.

For the graph $H$ and vertices $\alpha,\beta \in V(H)$, we define gadgets $\cJ_L=(J_L,\tau_L)$ and $\cJ_R=(J_R,\tau_R)$ as follows: Graphs $J_L,J_R$ are just edges $sx$ and $ty$, respectively. The pinning functions are given by $\tau_L(x)=\alpha$, $\tau_R(y)=\beta$. 

The next lemma follows straightforwardly from the definitions and guarantees that these gadgets satisfy items \ref{p:expansion} and \ref{p:node_map} of Definition~\ref{def::hardness} (note that (1) is a direct implication of (3)).

\begin{lemma}\label{lem:neighNode}
For graph $H$, vertices $\alpha,\beta \in V(H)$, and $\Dl_1=N_H(\alpha)$, $\Dl_2=N_H(\beta)$  the following hold
\begin{enumerate}
\item [(1)]
    if $\sigma \in\Hom(\cJ_L,H)$ then $\sigma(s) \in\Dl_1$, and 
    if $\sigma \in\Hom(\cJ_R,H)$ then $\sigma(t) \in\Dl_2$,
\item [(2)]
    for any $\gm_1\in \Dl_1$ and $\gm_2\in\Dl_2$, it holds that $|\Hom((\cJ_L,s),(H,\gm_1))| = |\Hom((\cJ_R,t),(H,\gm_2))|=1$,
\item [(3)]
    for any $\gm'_1\not\in\Dl_1$ and $\gm'_2\not\in\Dl_2$, it holds that $\Hom((\cJ_L,s),(H,\gm'_1))=\Hom((\cJ_R,t),(H,\gm'_2))= \emptyset$.
\end{enumerate}
\end{lemma}

The other type of a vertex gadget uses a cycle in $H$. 

Let $C=\theta \gamma_1 \gamma_2 \cdots \gamma_k \theta$ be a cycle in $H$ of length at least three. Gadgets $\cJ_{CL}=(J_{CL},\tau_{CL})$ and $\cJ_{CR}=(J_{CR},\tau_{CR})$ are defined as follows, see Fig.~\ref{fig:circular-vertex},
\begin{gather*}
    V(J_{CL})=\{s\} \cup \{ v_i, u_i : i \in [k] \} \cup \{ x \}, \\
    E(J_{CL})= \{v_iv_{i+1}: i \in [k-1] \} \cup \{v_iu_i : i \in [k] \} \cup \{sv_1, v_{k}s,sx\}.
\end{gather*}
The pinning function is given by $\tau(u_i)=\gamma_{i}$ for all $i \in [k]$ and $\tau (x)=\theta$.

\begin{figure}[ht]
    \centering
  \includegraphics[height=5cm]{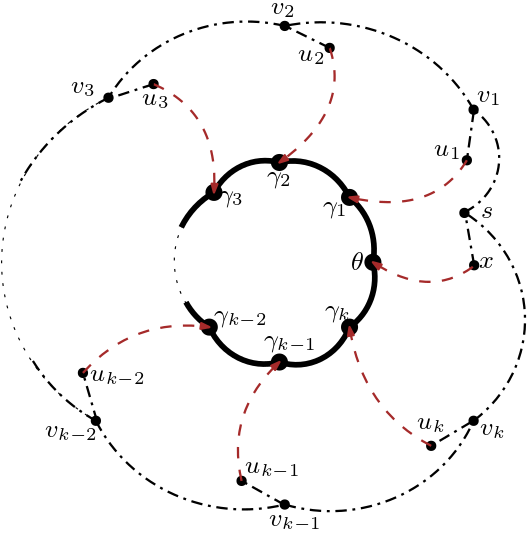}
  \caption{The vertex gadget $\cJ_{CL}$ based on the cycle $C=\theta \gamma_1 \gamma_2 \cdots \gamma_k \theta$. The edges of the gadget are shown by dash-dot lines, and the pinning function by dashed lines.}
    \label{fig:circular-vertex}
\end{figure}

The gadget $\cJ_{CR}$ is defined in the same way, except $s$ is replaced with $t$.

\begin{lemma}\label{lem:CycleNode}
For a square-free graph $H$, a cycle $C=\theta \gamma_1 \gamma_2 \cdots \gamma_k \theta$ in $H$ of length at least three, and $\Delta=\{ \gamma_1,\gamma_k \}$ the following hold
\begin{enumerate}
\item [(1)]
    if $\sigma \in\Hom(\cJ_{CL},H)$ or $\sg\in\Hom(\cJ_{CR},H)$, then $\sigma(s) \in \{\gamma_1,\gamma_k \}$ and $\sigma(t) \in \{\gamma_1,\gamma_k \}$, respectively;
\item [(2)]
    for any $\gamma\in \Delta$, 
    \[
    |\Hom((J_{CL},s),(H,\gamma))| =|\Hom((J_{CR},t),(H,\gamma))| = 1,
    \]
\item [(3)]
    for any $\gamma'\not \in \Delta$
    \[
    \Hom((\cJ_{CL},s),(H,\gamma')) =\Hom((\cJ_{CR},t),(H,\gamma'))= \emptyset.
    \]
\end{enumerate}
\end{lemma}

\begin{proof}
For item (1) the cycle $C$ is an nc-walk, therefore Lemma~\ref{lem:shiftingWalk} can be applied. Clearly, $\sg(v)\in N_H(\th)$ for any homomorphism $\sigma \in\Hom(\cJ_{CL},H)$. Suppose that there exists $\sigma \in\Hom(\cJ_{CL},H)$ such that $\sigma(s)=\alpha\in N_H(\theta)\setminus \{\gamma_1,\gamma_k\}$. Then $\sigma \in\Hom((\cJ_{CL},s),(H,\alpha))$, hence by Lemma~\ref{lem:shiftingWalk} $\sigma(v_{i})=\gamma_{i-1}$ for all $i \in [k]$. Also, in a similar way $\sigma(v_{i})=\gamma_{i+1}$ for all $i \in [k]$, a contradiction. 

A proof for $\cJ_{CR}$ is analogous.

For item (2) without loss of generality assume that $\sigma\in\Hom((\cJ_{CL},s),(H,\gamma_1))$, or in other words $\sg(s)=\gm_1$. Since $C$ is an nc-walk by Lemma~\ref{lem:shiftingWalk} it holds $\sigma(v_i)=\gamma_{i+1}$ for all $i\in[k]$, we just need to set $\gamma_{k+1}=\th$. Therefore $\sigma$ is determined uniquely. The same argument works if $\sigma\in\Hom((\cJ_{CL},s),(H,\gamma_k))$, and for $\cJ_{CR}$.

Item (3) follows straightforwardly by Lemma~\ref{lem:CycleNode}(1).
\end{proof}

% SHARP P - COMPLETNESS ============================================================================
\section{The hardness of $\parthom pH$}\label{sec:cases}

In this section we prove the hardness part of Theorem~\ref{the:main-intro}. More specifically we will apply Theorem~\ref{TH:hardness} and the constructions from Section~\ref{sec:gadgets} to show that $\BIS{\ld_1}{\ld_2}$ is Turing reducible to $\parthom pH$.

We consider three cases depending on the existence of vertices of certain degree in $H$. In each of the three cases we use slightly different variations of vertex and edge gadgets.

\smallskip

\noindent
\textsc{Case 1.} The graph $H$ has at least two vertices $\alpha$ and $\beta$ such that $\deg(\alpha),\deg(\beta) \not\equiv 1 \pmod{p}$.

\smallskip

Let $S=\{ \gamma \in V(H) : \deg(\gm) \not \equiv 1 \mod{p} \}$; we know that $S$ contains at least two elements. Pick $\alpha,\beta\in S$ such that the distance between them is minimal. Let $W=\alpha \gamma_1 \cdots \gamma_{k-1} \gamma_k \beta$ be a shortest path between $\alpha,\beta$. By the choice of $W$, $\deg(\gamma_i)\equiv 1 \pmod{p}$ for all $i \in [k]$.  

We make an edge gadget $\cK=(K,\tau)$ for this case based on this path as defined in Section~\ref{sec:edge-gadget}. More precisely,
\begin{gather*}
    V(K)=\{ s,t \} \cup \{ v_i, u_i : i \in [k] \} , \\
    E(K)= \{ v_iv_{i+1}: i \in [k-1] \} \cup \{v_iu_i : i \in [k] \} \cup \{sv_1,v_{k}t\}.
\end{gather*}
The pinning function is given by $\tau (u_i)=\gamma_i$ for all $i \in [k]$.

Any path is a nc-walk, so we can apply Lemma~\ref{lem:coungingWalk} to $W$. For the gadgets we use $\Dl_1=N_H(\alpha),\Dl_2=N_H(\beta)$ and $\dl_1=\gm_1,\dl_2=\gm_k$. This satisfies property \ref{p:size} of hardness gadgets, because $\deg(\alpha),\deg(\beta) \not\equiv 1 \pmod{p}$. Then for any $\om_s \in \Dl_1-\dl_1$ and $\om_t \in \Dl-\dl_2$  we have
\begin{enumerate}
\item [(1)]
    $|\Hom((\cK,s,t),(H,\om_s,\om_t))|=0$;
\item [(2)]
    $|\Hom((\cK,s,t),(H,\gm_1,\om_t))| \equiv 1 \pmod{p}$;
\item [(3)]
    $|\Hom((\cK,s,t),(H,\om_s,\gamma_{k}))| \equiv 1 \pmod{p}$
\end{enumerate}
Also, for any $i \in [k]$ we have $\deg(\gamma_i)\equiv 1 \pmod{p}$, and so
\begin{equation*}
|\Hom((\cK,s,t),(H,\gamma_1,\gamma_{k}))| = 1 + \sum_{i=1}^{k} (\deg(\gamma_i)-1)= 1 + 0 \equiv 1 \pmod{p}.
\end{equation*}
Hence,
\begin{enumerate}
\item [(4)]
    $|\Hom((\cK,s,t),(H,\gamma_1,\gamma_{k}))| \equiv 1 \pmod{p}$.
\end{enumerate}

Thus $\cK$ satisfies properties \ref{p:edge_map1}, \ref{p:edge_map2}, \ref{p:edge_map3}, and \ref{p:edge_map4} of  hardness gadgets.

For a vertex gadget we use the first type, that is $\cJ_L,\cJ_R$ are just edges $sx$ and $ty$, respectively, see Fig.~\ref{fig:case1}. 
By Lemma \ref{lem:neighNode} these gadgets satisfy properties \ref{p:expansion} and \ref{p:node_map} of  hardness gadgets.  

Thus Theorem \ref{TH:hardness} yields a required reduction.

\begin{figure}[H]
    \centering
    \includegraphics[height=4cm]{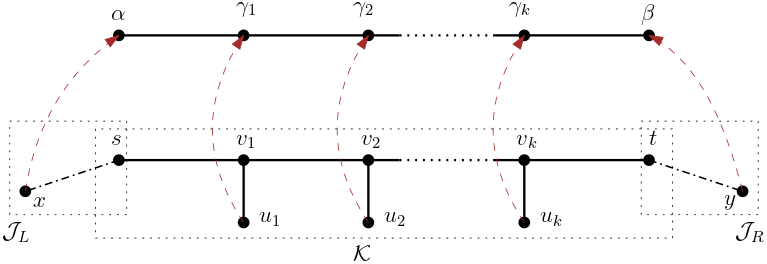}
    \caption{Vertex and edge gadgets based on the path $W=\alpha \gamma_1 \cdots \gamma_k \beta$. The vertex gadgets $\cJ_L$ and $\cJ_R$ are shown as dot-dashed boxes. The pinning function is shown by dashed lines.}
    \label{fig:case1}
\end{figure}

\smallskip

\noindent
\textsc{Case 2.} Graph $H$ has exactly one vertex $\theta$ such that $\deg(\theta) \not\equiv 1 \pmod{p}$.

\smallskip

In this case we further split into two subcases. However, before we proceed with that we rule out the case of trees.

\begin{lemma}\label{lem:treeNodes}
Let $H$ be a tree that has no automorphism of order $p$ and is not a star. Then $H$ has at least two vertices $\alpha$ and $\beta$ such that $\deg(\alpha), \deg(\beta) \not \equiv 1 \pmod{p}$.
\end{lemma}

\begin{proof}
Let $P=v_0 v_1 \cdots v_{l-1} v_l$ be a maximal path in $H$, $l$ is the length of $P$. Since $H$ is not a star, $l>2$.

First, observe that $v_0,v_l$ must be leaves, because otherwise $\deg(v_0)>1$ or $\deg(v_l)>1$ and $P$ can be extended in at least one direction.

Next, we show that $\deg(v_1),\deg(v_{l-1})\not \equiv 1 \pmod{p}$. Indeed,  if $\deg(v_1)\equiv 1 \pmod{p}$, then $\deg(v_1)>p$ because $v_1$ is not a leaf itself.  Therefore $N_H(v_1)-v_2=\{ w_0, w_1, \cdots , w_{kp-1}\}$ for some $k>0$ consists of leaves. Set $L=\{w_0,\dots,w_{p-1}\}$ and define $\sigma$ to be the mapping from $H$ to itself given by
\begin{equation*}
    \sigma(w) = 
     \begin{cases}
       w &\quad  \text{if }w\in V(H)\setminus L\\
       w_{i+1 \!\!\!\!\! \pmod{p}} &\quad \text{if } w=w_i \;\; \text{such that} \;\; w_i\in L. 
     \end{cases}
\end{equation*}
Then, as is easily seen, $\sg$ is an automorphism of $H$ of order $p$. Indeed, for any edge $xy\in E(H)$, if $x,y \not\in L$, then $\sigma(x)\sigma(y)=xy\in E(H)$; if $x=w_i\in L$, then $y=v_1$, thus $\sigma(w_i)\sigma(v_1)=w_{i+1 \!\! \pmod{p}}v_1\in E(H)$; finally, both $x$ and $y$ cannot belong to $L$. It is a contradiction with the assumptions on $H$. Hence, the degrees of $v_1$ and $v_{l-1}$ are not equal to $1\pmod p$.
\end{proof}

Thus, we may assume that $H$ is not a tree.

\smallskip
\noindent
\textsc{Case 2.1.} The vertex $\theta$, $\deg(\th)\not\equiv1\pmod p$, is on a cycle $C$.

\smallskip

In this case the edge gadget is based on the cycle $C$. More precisely, 
let $C=\theta \gamma_1 \gamma_2 \cdots \gamma_k \theta$ be a cycle in $H$ of length at least 3 and such that for all $i \in [k]$ it holds that $\deg(v_i)\equiv 1 \pmod{p}$ and $\deg(\theta) \not \equiv 1 \pmod{p}$. We define gadget $\cK=(K,\tau)$ as follows:\\
--  $V(K)=\{ s,t \} \cup \{ v_i, u_i : i \in [k] \}$; \\
--  $E(K)= \{ \edge{v_i}{v_{i+1}}: i \in [k-1] \} \cup \{\edge{v_i}{u_i} : i \in [k] \} \cup \{ \edge{s}{v_1},\edge{v_{k}}{t} \}$;\\
-- the labeling function is given by $\tau (u_i)=\gamma_{i}$ for all $i \in [k]$.\\[3mm]

Set  $\Dl_1=\Dl_2=N_H(\th)$ and $\dl_1=\gm_1,\dl_2=\gm_k$. These parameters satisfy property \ref{p:size} of a hardness gadget, because $\deg(\theta) \not\equiv 1 \pmod{p}$.  A cycle is an nc-walk, so we can apply Lemma~\ref{lem:coungingWalk} to obtain the following

\begin{lemma}\label{lem:countingCycle}
Let $H$ be a square-free graph and $\cK$ an edge gadget based on the cycle $C=\theta \gamma_1 \gamma_2 \cdots \gamma_k \theta$ in $H$. For any $\om_s \in \Dl_1-\dl_1$ and $\om_t \in \Dl_2-\dl_2$,
\begin{enumerate}
\item [(1)]
    $|\Hom((\cK,s,t),(H,\om_s,\om_t))|=0$;
\item [(2)]
    $|\Hom((\cK,s,t),(H,\dl_1,\om_t))| \equiv 1 \pmod{p}$;
\item [(3)]
    $|Hom((\cK,s,t),(H,\om_s,\dl_2))| \equiv 1 \pmod{p}$;
\item [(4)]
    $|Hom((\cK,s,t),(H,\dl_1,\dl_2))|  \equiv 1 \pmod{p}$.
\end{enumerate}
\end{lemma} 

\begin{proof}
The cycle $C$ is an nc-walk. Therefore by Lemma~\ref{lem:coungingWalk} items (1), (2), and (3) hold. For item (4) note that $\deg(\gamma_i)\equiv 1 \pmod{p}$ for all $i \in [k]$, therefore
\begin{equation*}
    |\Hom((\cK,s,t),(H,\gamma_1,\gamma_{k}))| = 1 + \sum_{i=1}^{k} (\deg(\gamma_i)-1)= 1 + 0 \equiv 1 \pmod{p}.
\end{equation*}
\end{proof}

By Lemma~\ref{lem:countingCycle} gadget $\cK$ satisfies properties \ref{p:edge_map1}, \ref{p:edge_map2}, \ref{p:edge_map3}, and \ref{p:edge_map4} of hardness gadgets.

For vertex gadgets we take $\cJ_L,\cJ_R$ (which are just edges) defined in Section~\ref{sec:node-gadget}, with $\alpha=\beta=\th$, see Fig~\ref{fig:case2.1}.  By Lemma \ref{lem:neighNode}, these gadgets satisfy properties \ref{p:expansion} and \ref{p:node_map} of hardness gadgets.  

Finally, by Theorem~\ref{TH:hardness} $\BIS{\ld_1}{\ld_2}$ is Turing reducible to $\parthom pH$.

\begin{figure}[ht]
    \centering
    \includegraphics[height=5cm]{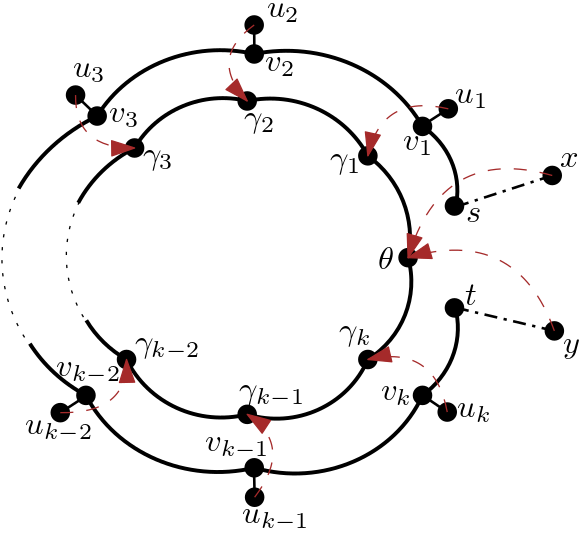}
    \caption{Hardness gadgets corresponding to the cycle $C=\theta \gamma_1 \cdots \gamma_k \theta$. The vertex gadgets $\cJ_L$ and $\cJ_R$ are shown by dot-dashed lines. The pinning function is shown by dashed lines.}
    \label{fig:case2.1}
\end{figure}

\smallskip
\noindent
\textsc{Case 2.2.} The vertex $\theta$ is not on any cycle.

\smallskip

Since $H$ is not a tree, it contains at least one cycle; let $C$ be such a cycle. Let $P=\gamma_0 \gamma_{k+1} \gamma_{k+2} \cdots \gamma_{k+k'} \theta$ be a shortest path from a vertex $\gamma_0$ on cycle $C=\gamma_0 \gamma_1 \gamma_2 \cdots \gamma_k \gamma_0$, $k\ge2$, to $\theta$. Note that $\deg(\gamma_i) \equiv 1 \pmod{p}$ for all $\gm_i$, $i\in \{0,\dots,k+k'\}$. Edge gadget $\cK$ in this case is based on the walk
$$
W=\theta \gamma_{k+k'} \cdots \gamma_{k+2} \gamma_{k+1} \gamma_0 \gamma_1 \gamma_2 \cdots \gamma_k\gm_0 \gm_{k+1} \gamma_{k+2} \cdots \gamma_{k+k'} \theta.
$$
Note that $W$ is an nc-walk. More precisely, 
the gadget $\cK=(K,\tau)$ is defined as follows:\\
--  $V(K)=\{ s,t \} \cup \{ v_i, u_i : i \in [k+2k'+2] \}$;\\
--  $E(K)= \{ v_iv_{i+1}: i \in [k+2k'+1] \} \cup \{ v_iu_i : i \in [k+2k'+2] \} \cup \{sv_1,v_{k+2k'+2}t\}$;\\
-- the pinning function is given by
\begin{equation*}
    \tau(u_i) = 
     \begin{cases}
       \gamma_{k+k'+1-i} &\quad 1 \leq i \leq k',\\
       \gamma_{i-k'-1} &\quad k'+1 \leq i \leq k+k'+1, \\
       \gamma_0 &\quad i=k+k'+2,\\
       \gamma_{i-k'-2} & \quad k+k'+3 \leq i \leq k+2k'+2.\\ 
     \end{cases}
\end{equation*}

Set $\delta_1=\delta_2= \gamma_{k+k'}$ and $\Delta_1 =\Delta_2= N_H(\theta)$. These parameters satisfy property \ref{p:size} of hardness gadgets, because $\deg(\theta) \not\equiv 1 \pmod{p}$. As $W$ is an nc-walk, by Lemma~\ref{lem:coungingWalk} for any $\om\in N_H (\theta)-\gamma_{k+k'}$, we have
\begin{enumerate}
\item [(1)]
    $|\Hom((\cK,s,t),(H,\om,\om))|=0$;
\item [(2)]
    $|\Hom((\cK,s,t),(H,\gamma_{k+k'},\om))|\equiv 1 \pmod{p}$;
\item [(3)]
    $|\Hom((\cK,s,t),(H,\om,\gamma_{k+k'}))|\equiv 1 \pmod{p}$.
\end{enumerate}
Also, $\deg(\gamma_i)\equiv 1 \pmod{p}$ for all $i \in [k+k']\cup \{0 \}$. Therefore
\begin{equation*}
    |\Hom((\cK,s,t),(H,\gamma_{k+k'},\gamma_{k+k'}))| = 1 + \sum_{i=1}^{k+k'} (\deg(\gamma_i)-1)= 1 + 0 \equiv 1 \pmod{p}.
\end{equation*}
Hence,
\begin{enumerate}
    \item [(4)]
    $|\Hom((\cK,s,t),(H,\gamma_{k+k'},\gamma_{k+k'}))| \equiv 1 \pmod{p}$.
\end{enumerate}
Thus the gadget $\cK$ satisfies properties \ref{p:edge_map1}, \ref{p:edge_map2}, \ref{p:edge_map3}, and \ref{p:edge_map4} of hardness gadgets.

Finally, for vertex gadgets we again use gadgets $\cJ_L,\cJ_R$ introduced in Section~\ref{sec:node-gadget}, with $\alpha=\beta=\th$, see Fig~\ref{fig:case2.2}. By Lemma~\ref{lem:neighNode}, these gadgets satisfy properties \ref{p:expansion} and \ref{p:node_map} of hardness gadgets. Thus, by Theorem~\ref{TH:hardness} $\BIS{\ld_1}{\ld_2}$, $\ld_1=\ld_2=|N_H(\th)|-1$ is Turing reducible to $\parthom pH$. 

\begin{figure}[ht]
    \centering
  \includegraphics[height=6cm]{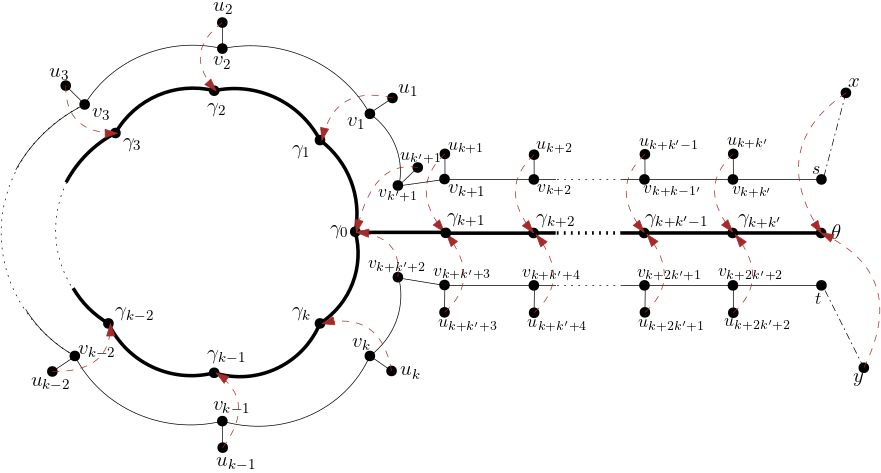}
  \caption{Gadget $\cK$ based on the nc-walk $W=\theta \gamma_{k+k'} \cdots \gamma_{k+2} \gamma_{k+1} \gamma_0\gamma_1 \gamma_2 \cdots \gamma_k\gm_0 \gamma_{k+1} \gamma_{k+2} \cdots \gamma_{k+k'} \theta$. The vertex gadgets $\cJ_L$ and $\cJ_R$ corresponding to $\theta$ are shown by dot-dashed lines. The pinning function is shown by dashed lines.}
    \label{fig:case2.2}
\end{figure}

\smallskip
\noindent
\textsc{Case 3.} For every vertex $\gamma\in V(H)$ it holds $\deg(\gamma) \equiv 1 \pmod{p}$.

\smallskip

By Lemma \ref{lem:treeNodes}, $H$ is not a tree, therefore it contains a cycle $C=\theta \gm_1 \gm_2 \cdots \gm_k \theta$ such that $k\ge2$. Set $\delta_1=\gamma_1 ,\delta_2= \gamma_k$ and $\Delta_1 =\Delta_2= \{\gamma_1, \gamma_k \}$. These parameters satisfy property \ref{p:size} of hardness gadgets, because $|\Dl_1|=|\Dl_2|\not\equiv 1 \pmod{p}$. An edge gadget $\cK$ is based on this cycle $C$ as in Case~2.1. More precisely, we define gadget $\cK=(K,\tau)$ as follows:\\
--  $V(K)=\{ s,t \} \cup \{ v_i, u_i : i \in [k] \}$; \\
--  $E(K)= \{ \edge{v_i}{v_{i+1}}: i \in [k-1] \} \cup \{\edge{v_i}{u_i} : i \in [k] \} \cup \{ \edge{s}{v_1},\edge{v_{k}}{t} \}$;\\
-- the labeling function is given by $\tau (u_i)=\gamma_{i}$ for all $i \in [k]$.\\[3mm]

A cycle is an nc-walk, so as in Case~2.1 we can apply Lemma~\ref{lem:coungingWalk} to obtain the following

\begin{lemma}\label{lem:countingCycle3}
Let $H$ be a square-free graph and $\cK$ an edge gadget based on the cycle $C=\theta \gamma_1 \gamma_2 \cdots \gamma_k \theta$, $k\ge2$, in $H$. For any $\om_s \in \Dl_1-\dl_1$ and $\om_t \in \Dl_2-\dl_2$,
\begin{enumerate}
\item [(1)]
    $|\Hom((\cK,s,t),(H,\om_s,\om_t))|=0$;
\item [(2)]
    $|\Hom((\cK,s,t),(H,\dl_1,\om_t))| \equiv 1 \pmod{p}$;
\item [(3)]
    $|Hom((\cK,s,t),(H,\om_s,\dl_2))| \equiv 1 \pmod{p}$;
\item [(4)]
    $|Hom((\cK,s,t),(H,\dl_1,\dl_2))|  \equiv 1 \pmod{p}$.
\end{enumerate}
\end{lemma}  

Since $\deg(\gm) \equiv 1 \pmod{p}$ for every $\gamma\in V(H)$, By Lemma~\ref{lem:countingCycle3} $\cK$ satisfies properties \ref{p:edge_map1}, \ref{p:edge_map2}, \ref{p:edge_map3}, and \ref{p:edge_map4} of hardness gadgets. 

For vertex gadgets we choose $\cJ_{CL},\cJ_{CR}$ defined in Section~\ref{sec:node-gadget}, see Fig~\ref{fig:case3}. More precisely, gadgets $\cJ_{CL}=(J_{CL},\tau_{CL})$ and $\cJ_{CR}=(J_{CR},\tau_{CR})$ are defined as follows, see Fig.~\ref{fig:circular-vertex},
\begin{gather*}
    V(J_{CL})=\{s\} \cup \{ v_i, u_i : i \in [k] \} \cup \{ x \}, \\
    E(J_{CL})= \{v_iv_{i+1}: i \in [k-1] \} \cup \{v_iu_i : i \in [k] \} \cup \{sv_1, v_{k}s,sx\}.
\end{gather*}
The pinning function is given by $\tau(u_i)=\gamma_{i}$ for all $i \in [k]$ and $\tau (x)=\theta$ (for $\cJ_{CR}$, $\tau(y)=\theta$). Gadget $\cJ_{CR}$ is defined the same way with replacement of $s$ with $t$ and $x$ with $y$. By Lemma \ref{lem:CycleNode}, these gadgets satisfy properties \ref{p:expansion} and \ref{p:node_map} of hardness gadgets. 
Therefore, by Theorem \ref{TH:hardness}, $\BIS{\ld_1}{\ld_2}$, $\ld_1=\ld_2=|\{\gm_1,\gm_k\}|-1=1$ is Turing reducible to $\parthom pH$. 

\begin{figure}[ht]
    \centering
    \includegraphics[height=5cm]{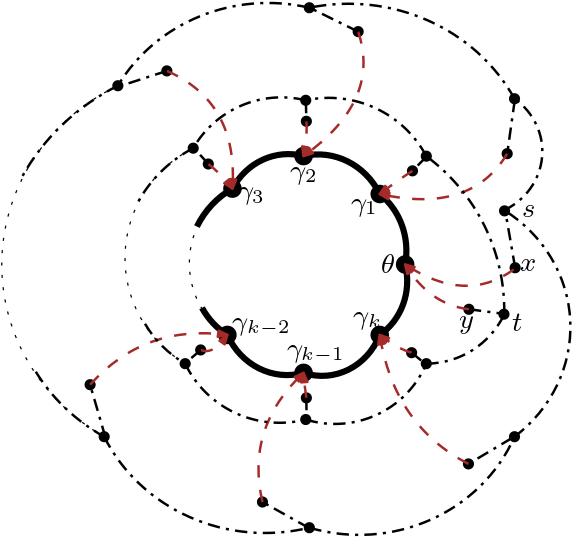}\hspace{1cm}
    \includegraphics[height=5cm]{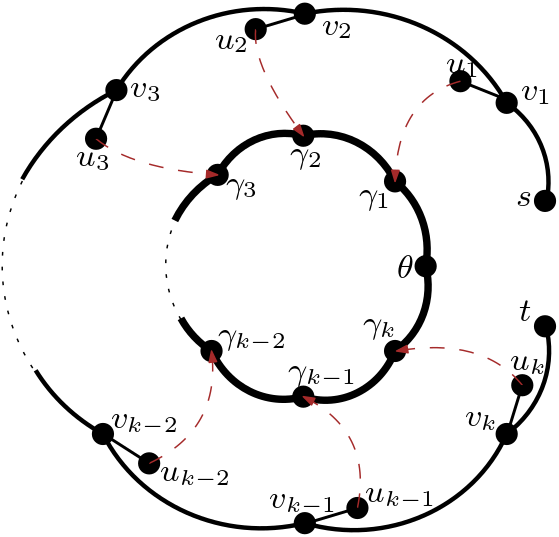}
    \caption{Hardness gadgets based on cycle $C=\theta \gamma_1 \cdots \gamma_k \theta$.
On the left are the vertex gadgets $\cJ_{CL}$ and $\cJ_{CR}$  shown by dot-dashed lines with $J_{CR}$ inside $\cJ_{CL}$. $\cJ_{CL}$ is the cycle containing vertex $s$, and $\cJ_{CR}$ is the cycle containing vertex $t$. The remaining vertices of the gadgets are not labelled. The pinning function is shown by dashed lines.
On the right, the edge gadget $\cK$ is highlighted. Again, the pinning function is represented by dashed lines.}
    \label{fig:case3}
\end{figure}

\bibliographystyle{plain}

\bibliography{modular}
\end{document}